\documentclass{article}
\usepackage{ijcai19}

\usepackage{times}
\usepackage{soul}
\usepackage{url}
\usepackage[utf8]{inputenc}
\usepackage[small]{caption}
\usepackage{graphicx}
\usepackage{amsmath}
\usepackage{booktabs}
\usepackage{helvet}
\usepackage{courier}
\usepackage{graphicx}
\usepackage{amsmath}
\usepackage{threeparttable}
\usepackage{amsthm}
\usepackage{color}
\usepackage{amsfonts}
\usepackage{mathrsfs}
\usepackage{framed}
\usepackage{flushend}
\usepackage[ruled,vlined,linesnumbered]{algorithm2e}
\title{Diffusion and Auction on Graphs}

\author{
	Bin Li$^1$
	\and
	Dong Hao$^1$\footnote{Dong Hao is the corresponding author.}\and
	Dengji Zhao$^2$
	\And
	Makoto Yokoo$^3$
	\affiliations
	$^1$University of Electronic Science and Technology of China\\
	$^2$ShanghaiTech University $\quad$ $^3$Kyushu University\\
}

\begin{document}

\maketitle
\begin{abstract}
	Auction is the common paradigm for resource allocation which is a fundamental problem in human society. Existing research indicates that the two primary objectives, the seller's revenue and the allocation efficiency, are generally conflicting in auction design. For the first time, we expand the domain of the classic auction to a social graph and formally identify a new class of auction mechanisms on graphs. All mechanisms in this class are incentive-compatible and also promote all buyers to diffuse the auction information to others, whereby both the seller's revenue and the allocation efficiency are significantly improved comparing with the Vickrey auction. It is found that the recently proposed information diffusion mechanism is an extreme case with the lowest revenue in this new class. Our work could potentially inspire a new perspective for the efficient and optimal auction design and could be applied into the prevalent online social and economic networks.
\end{abstract}
\let\thefootnote\relax\footnotetext{Email address: libin@std.uestc.edu.cn, haodong@uestc.edu.cn, zhaodj@shanghaitech.edu.cn, yokoo@inf.kyushu-u.ac.jp}
\section{Introduction}
Auction has been a common paradigm for allocating resources \cite{krishna2009auction}, its applications vary from assigning the sponsored links to advertisers \cite{edelman2007internet} to allocating the wireless spectrum worldwide \cite{cramton2013spectrum}. Auctions ask and answer this question: who should get the goods and at what prices? Typically, there are two goals in auction design. One is to design auctions that maximize social welfare, i.e., allocating the goods in a way that maximizes the total utilities of all participants. Another line focuses on the seller's revenue. Traditionally, there are two ways to improve the revenue: use the classic optimal auction \cite{Myerson1981OptimalAD} or gather more people to join the auction. The former method requires prior knowledge of buyers' valuations to compute a reserve price. The buyer set needs to be fixed in advance, and the mechanism will be invalid if a seller can hardly have a panorama (value distribution) of all the underlying bidders. As a comparison, the latter method seems to be more effective and adaptive comparing with optimal auctions. A classic result in \cite{Bulow1996Auction} shows that when the buyers' valuations draw i.i.d from a regular distribution, then the expected revenue of a second price auction with one additional bidder is no less than that given by the optimal auction. However, except using costly advertising, \emph{where these additional buyers could come from?}


The natural way is to ask the current bidders to help to diffuse the sale information to the outsiders.
However, the selfish bidders have no incentives to do so. From the point of view of the whole economic network, this \emph{information blocking} not only decreases the seller's revenue, but also reduces the allocation efficiency of goods. Here are two examples. On online selling websites like eBay, the final price of an item depends on the population of interested buyers. Normally, the buyers have no incentives to invite other competitors and the information are limited to people who see this sale. One activity on online social networks that is frequently seen is people sharing clicking links (such as a new product) to their friends. This is actually an informal way to use information diffusion to recruit more people to a sale. Essentially, the above prevalent problems are about how to design a good market for the sale of products or services where the participants are encouraged to invite more people to the mechanism.
The problem can be seen as auctions constrained with (negative) externalities \cite{bhattacharya2011allocations,haghpanah2013optimal,belloni2017mechanism}. However, for our problem, the underlying social connections of all participants gives the externality constrains in a very complicate way, where previous common techniques on externalities can hardly help.
The most related works are from \cite{li2017mechanism,Li2018CustomerSI,Zhao2018Multi}. \cite{li2017mechanism} initiates this problem and proposes one mechanism which incentivizes buyers to diffuse seller's sale information to their neighbors. \cite{Zhao2018Multi} extends this to a multi-unit setting where the seller has multiple homogeneous items for sale and each buyer wants at most one of them. \cite{Li2018CustomerSI} further investigates this problem in constrained economic networks where both the allocation efficiency and revenue can be obtained.

We study the problem of auction on graphs from a more general point of view.
In the first part of this paper, we propose a class of mechanisms named \emph{critical diffusion mechanism (CDM)} on unweighted graphs. We shows that the method proposed in \cite{li2017mechanism} is a special case of this class and is the one with the lowest revenue. In the second part, we further study this problem on weighted graphs which has never been tackled in the literature.
We propose the very first solution called \emph{weighted diffusion mechanism (WDM)}. In both CDM and WDM, each buyer's optimal choice is to report her valuation truthfully and diffuse the sale information to all her neighbors. Our theoretical results can be applied to the research and application fields such as crowdsourcing, sharing economics and viral marketing. By using our mechanisms, the buyers on eBay are happy to do the invitations, the diffusion of an innovation will be maximized as well.
\section{Preliminaries}
\newtheorem{defn}{Definition}
\newtheorem{prop}{Proposition}
\newtheorem{lemma}{Lemma}
\newtheorem{theorem}{Theorem}
\newtheorem{corollary}{Corollary}
\newtheorem{exmp}{Example}
Consider a seller selling an item in a digraph $G=(V,E)$ where $V$ is the node set with $|V|=n$ and $E$ is the edge set. Besides the seller $s$, each node $i\in V$ owns a private type $t_i=(v_i,r_i)$, where $v_i$ is her valuation on the item and $r_i$ is the set of her neighboring nodes. Node $i$ can only communicate with $j$ through existing link $(i,j)\in E$, and a node could participate in the auction only if someone of her neighbors has joined in the auction and further informs her of the sale. Initially, only the seller's neighbors are informed of the sale.
Let $w(i,j)$ be the weight on edge $(i,j)$, for instance in a distribution network $w(i,j)$ can be represented as the freight of transporting the item from $i$ to $j$. We assume graph $G$ contains no negative cycles and $w(i,j)$ is known once $i$ and $j$ join in the auction.
In Figure \ref{summary}, node $s$ owns one item for sale and only node $A$ and $B$ are informed of the sale at first. Node $C$ cannot join in the auction if $B$ does not inform her of the sale. The values in each circle are nodes' private valuations and the red numbers represent for the resided weights.

We model the selling problem as an auction. Formally, denote $t_i=(v_i,r_i)$ by node $i$'s private type and $t=(t_1,\cdots,t_n)$ by the type profile of all nodes. Let $t_{-i}=\{t_1,t_2,\cdots,t_{i-1},t_{i+1},\cdots,t_n\}$ be the type profile of all other nodes except $i$, i.e., $t=(t_i,t_{-i})$. Let $T_i=\mathbb{R}_{\geq 0}\times \mathbb{P}(V)$ be the type space of $i$ where $\mathbb{P}(V)$ is the power set of $V$ and $T=\times T_{i\in V\setminus\{s\}}$ be the type profile space of all nodes.
As usual, let $t_i'=(v'_i,r'_i)\in T_i$ be the reported type of node $i$ where $r'_i\subseteq r_i$ means that $i$ diffuses the sale information to nodes in $r'_i$. The set $r'_i$ is limited to $\mathbb{P}({r_i})$ as node $i$ is not aware of other nodes who are not her neighbors. Denote $t'=(t_1',\cdots,t_n')$ by the reported type profile of all nodes where $t'_i=nil$ when $i$ has never been informed of the sale or $i$ has no interest in the auction.

\begin{defn}
	Given a reported type profile $t'$ and a node $i$ with $t_i'\neq nil$, define a \textit{trading path} from seller $s$ to $i$ as
	an ordered sequence of nodes $(a_1,a_2,\cdots,a_l,a_{l+1}=i)$ such that $a_1\in r_s$ and for $1<j\leq l+1$, $a_j\in r_{a_{j-1}}'$.
\end{defn}
Among all trading paths from the seller to node $i$, denote $L^*_i(t')$ by the shortest one, i.e., $L^*_i(t')=\arg \min_{L\in \mathbb{L}_i(t')}\sum_{(i,j)\in L}w(i,j)$ where $\mathbb{L}_i(t')$ is the set of all possible trading paths of node $i$. Since $G$ contains no negative cycles, then $L^*_i(t')$ is a simple path from $s$ to $i$.

\begin{defn}
	A \textit{diffusion auction} $\mathcal{M}=(\pi,x)$ defined on $G$ consists of two components: an \textit{allocation rule} $\pi:T\rightarrow \mathbb{L}$ which determines a single trading path $\pi(t)$, where $\mathbb{L}$ is the set of all possible trading paths in $G$ and $\pi(t)$ is a selected path whose terminal node will be allocated with the item; and a \textit{payment rule} $x=\{x_i\}_{i\in V\setminus\{s\}}$ which is the amount that each node pays, where $x_i:T\rightarrow \mathbb{R}$ is the payment rule for $i$.
\end{defn}

Given an allocation $\pi(t)$, the \textit{social welfare} in diffusion auctions is defined as $v(\pi(t))-\sum_{(i,j)\in \pi(t) }w(i,j)$ where $v(\pi(t))$ is winner's true valuation. Here the total weights $\sum_{(i,j)\in \pi'(t) }w(i,j)$ is perceived as the externalities incurred when passing an item from the seller to the winner. An allocation $\pi^*$ is efficient if it maximizes the social welfare for every $t\in T$, i.e., $\pi^* \in {\arg\max}_{\pi' \in \Pi} v(\pi'(t))-\sum_{(i,j)\in \pi'(t) }w(i,j)$
where $\Pi$ is the feasible allocation set. Denote $W^*(t)$ by the social welfare under efficient allocation $\pi^*$.

In our model, each node has quasilinear utility function, meaning that given $i$'s type $t_i$ and an allocation $\pi(t')$, $i$'s utility is defined as: $u_i(t_i, t', (\pi,x)) = v_iz_i(t')- x_i(t')$ where $z_i(t')$ is an indicator variable which is $1$ if $i$ wins the item and $0$ otherwise. A diffusion mechanism is \textit{individually rational} if for each node, her utility is non-negative when she reports her valuation truthfully regardless of her diffusion strategies, i.e., $u_i(t_i, ((v_i,r'_i),t'_{-i}), (\pi,x))\geq 0$.
And it is \textit{incentive-compatible} (IC) if reporting true types is a dominant strategy for every node in the auction, that is
$u_i(t_i, (t_i,t_{-i}^\prime), (\pi,x)) \geq u_i(t_i, (t_i^\prime, t_{-i}^{\prime\prime}), (\pi,x))$ for all $i\in V\setminus\{s\}$. Note that on the right side of the inequality, we use $t_{-i}^{\prime\prime}$ to replace $t_{-i}^{\prime}$ because the set of nil type nodes may change when $i$'s report becomes $t'_i$.
Given a reported type profile $t^\prime$ and a mechanism $\mathcal{M} = (\pi, x)$, the seller's \textit{revenue} generated by $\mathcal{M}$ is defined by $Rev^{\mathcal{M}}(t^\prime) = \sum_{i\in V\setminus\{s\}} x_i(t^\prime)-\sum_{(i,j)\in \pi(t')}w(i,j)$.

 	\begin{figure}[t]
 	\centering
 	\includegraphics[width=2.2in]{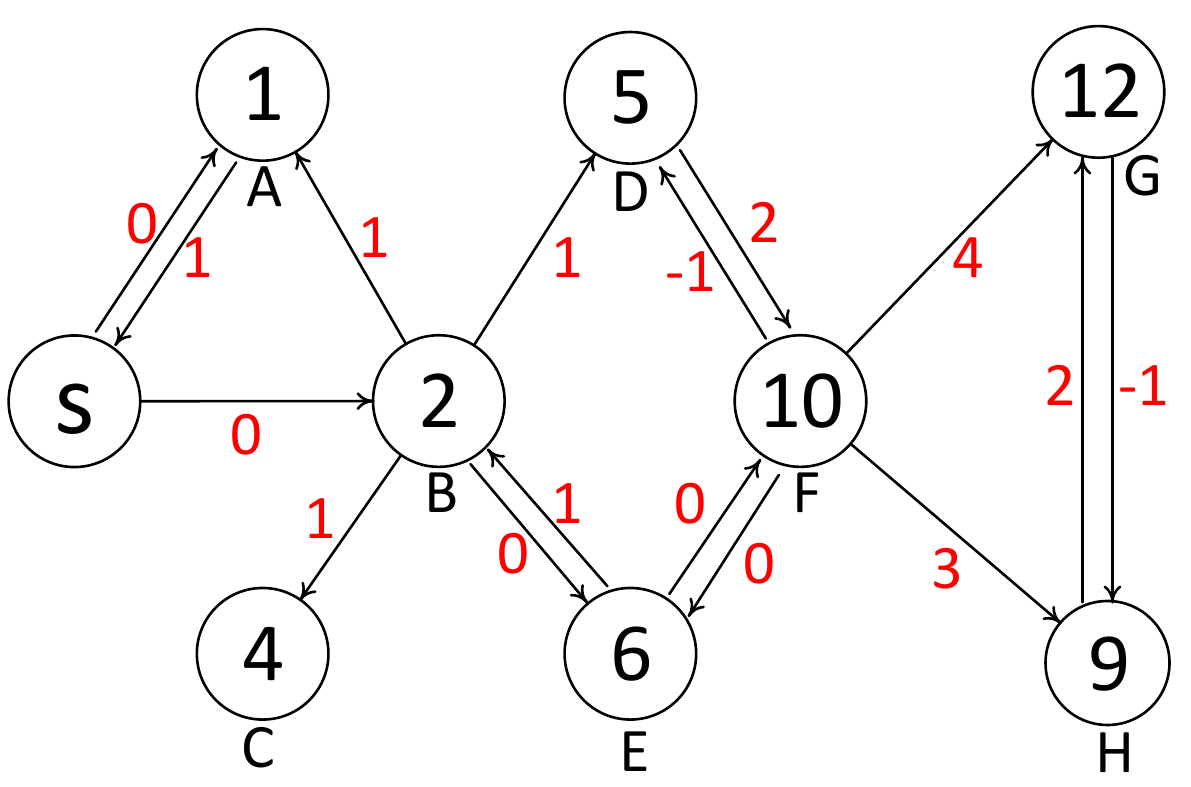}
 	\caption{An example of weighted graphs where $C^*_G(t)=\{B,F,G\}$ and $L^*_F(t)=\{B,E,F\}$.}\label{summary}
 \end{figure}

The VCG mechanism \cite{vickrey1961counterspeculation,clarke1971multipart,groves1973incentives} is a generic truthful mechanism for achieving an efficient allocation. However, it has many serious practical problems \cite{rothkopf2007thirteen}. Under our framework, it can lead to a high deficit for the seller \cite{li2017mechanism}. Instead, we use the outcomes of the Vickrey auction \cite{vickrey1961counterspeculation}, a special case of the VCG mechanism, as the benchmark and design good diffusion auctions that outperform them.
In the Vickrey auction, the item is allocated to the highest bidder who is charged with the second highest bid. Since every node in the auction has no incentives to share the sale information, then the allocation efficiency and the seller's revenue equal the highest and second highest bid in $\{t_i\}_{i\in r_s}$ respectively.



\section{Auction Mechanism on Unweighted Graph}
This section investigates diffusion mechanisms in unweighted graphs where each edge's weight is zero.
In an unweighted graph, for any determined node $i\in V\setminus\{s\}$, there would be multiple simple paths from $s$ to $i$.
Moreover, there are several cut nodes which are shared by all these paths. These nodes are called \textit{critical diffusion nodes} of $i$ since without any of them, node $i$ cannot receive the sale information and cannot join in the auction.

\begin{defn}\label{critical diffusion nodes}
	Given a reported type profile $t^\prime$ and a node $i$ with non-nil reported type,
	define $C_i(t')=\{\bigcap L\}_{L\in \mathbb{L}_i(t')}$ as the set of $i$'s critical diffusion nodes where $\mathbb{L}_i(t')$ is the set of all feasible trading paths from the seller $s$ to node $i$.
\end{defn}
Notice that for any pair $j,k\in C_i(t')$, we have that either $j\in C_k(t')$ or $k\in C_j(t')$. Therefore there is an unique fully ordered set $C^{*}_i(t')=\{s_1, s_2,\cdots,s_k, s_{k+1},\cdots,i\}$ for all nodes in $C_i(t')$ such that $s_{j}$ is a critical diffusion node of $s_{j+1}$. This unique sequence is named as \textit{the critical diffusion sequence of $i$}. In addition, for any $j\in C^*_i(t')$, it is clear that $C^*_j(t')=\{s_1,s_2,\cdots,j\}$. For example, there are four simple paths from the seller to node $G$ in Figure \ref{summary}, any of which has to pass node $B$ and $F$, and therefore node $G$'s critical diffusion nodes are $\{B,F,G\}$. Then, we have $C^*_G=\{B, F, G\}$ and $C^*_F=\{B, F\}$.
Let $d_i$ be the set of nodes whose critical diffusion nodes include $i$, e.g., $d_F=\{F,G,H\}$. Clearly, if $i$ is removed from the graph, the nodes in $d_i$ cannot join in the sale.
Denote $t'_{x}$ by the partially reported types from set $x$ and $t'_{-x}=t'\setminus t'_{x}$ as the reported type profile when set $x$ is removed from the graph, where $x$ could be a set of edges, vertices or a mixture. For an edge $(i,j)$, $t'_{-\{(i,j)\}}$ means that node $j$ is removed from $i$'s diffusion strategy $r_i'$ with respect to $t'$. For example, we have $t_{-\{F\}}=\{t_i\}_{i\in V\setminus d_F}$ and $t_{-\{(B,D),(B,E)\}}=\{t_i\}_{i\in \{A,B,C\}}$.
\begin{defn}\label{alpha}
	Given a reported type profile $t'$, assume that the highest bidder is $m$ and her critical diffusion sequence is $C^{*}_m(t')=\{1, 2,\cdots,k, k+1,\cdots,m\}$. Define $\alpha_m=\emptyset$ and for an arbitrary $i\in C^{*}_m(t')\setminus m$, predefine an edge set $\alpha_i=\{(j,l)\in E\}_{j\in d_i}$ which has the following properties:
	\begin{itemize}
		\item[1.] Information blocking: node $i+1\notin t'_{-\alpha_i}$, meaning that the nodes in $d_{i+1}$ cannot join in the auction if $\alpha_i$ is removed from the graph.
		\item[2.] Node independence: for any two reported type profiles $t^{'1}$ and $t^{'2}$ which only differ in $t'_{d_{i+1}}$, $\alpha^1_i=\alpha^2_i$. This property ensures that $\alpha_i$ is independent of the strategies of nodes in $d_{i+1}$.
		\item[3.] Diffusion monotonicity: if $r_i'\subseteq r_i''$, then $t'_{-\alpha'_i}\subseteq t'_{-\alpha''_i}$. That is, the set of non-nil type in $t'_{-\alpha_i}$ is monotonically increasing with $r_i'$.
	\end{itemize}
\end{defn}
The first property requires the set should be a cut. The second property says when the mechanism designer chooses $\alpha_i$, her choice should not be affected by what happened in $d_{i+1}$. The third property further requires the choice of $\alpha_i$ should not ruin the monotonicity of information diffusion.

Based on $\alpha_i$, we now give a class of new mechanisms named critical diffusion mechanisms (CDM) in Alg. \ref{CSM}. In CDM, only nodes in $C^{*}_m(t')$ are considered as the candidates of the winner. In the allocation policy, nodes ordered higher in $C^{*}_m(t')$ are given priorities to win. The algorithm computes whether a node has the ability to win one after another and stops once the first qualified node is identified. In the following theorems, we prove that such a "sequential" allocation rule combined with the payments defined in Alg. \ref{CSM} provides a class of auction mechanisms with remarkable performances.

\begin{algorithm}[t]
	\SetKwInOut{Input}{\textbf{Input}}\SetKwInOut{Output}{\textbf{Output}}
	
	\BlankLine
	initialize $\pi(t')=\emptyset$ and $\{x_i(t')=0\}_{i\in V\setminus\{s\}}$\;
	locate the highest bidder $m$, break tie arbitrarily\;
	compute $C^{*}_m(t')$ and denote it by $\{1, 2,\cdots, m\}$\;
	
	\For{$i\leftarrow 1$ \KwTo $m$}{
		compute $\alpha_i$\;
		\If{$v_i=W^*(t'_{-\alpha_i})$}{
			set $\pi(t')$ to be any trading path of $i$ and $x_i(t')=W^*(t'_{-i})$\;
			break\;
		}
		\Else{
			set $x_i(t')=W^*(t'_{-i})-W^*(t'_{-\alpha_i})$\;
		}
	}
	\caption{Critical Diffusion Mechanism (CDM)\label{CSM}}
\end{algorithm}

\begin{theorem}\label{cdm_outcome}
	The critical diffusion mechanism proposed in Alg. \ref{CSM} is individually rational and incentive-compatible.
\end{theorem}
\begin{proof}
	Assume node $g$ is the winner in Alg. \ref{CSM}. For any node $i\notin C^{*}_m(t')$, her utility is zero.
	The only way for $i$ to change her utility is to increase her bid and becomes the highest bidder. In this case, she will be the winner according to Alg. \ref{CSM} and will pay the previous highest bid $v_m'$ which is greater than her true value $v_i$.
	For any node $i\in C^*_g(t')\setminus\{g\}$, her utility is $W^*(t'_{-\alpha_i})-W^*(t'_{-i})$. The latter term is independent of node $i$, and according to diffusion monotonicity, the former term is maximized by choosing a diffusion strategy $r_i'=r_i$. In addition, becoming the winner is also a bad choice for $i$ since  $v_i-W^*(t'_{-i})<W^*(t'_{-\alpha_i})-W^*(t'_{-i})$. Regarding to the winner $g$, her utility is $v_g-W^*(t'_{-g})=W^*(t'_{-\alpha_g})-W^*(t'_{-g})$. Because of diffusion monotonicity, it is no good for $g$ to give up the chance of winning through lowering her bid. According to the first and second properties of $\alpha_i$, the winner is still node $g$ no matter what strategies nodes in $C^{*}_m(t')\setminus C^{*}_g(t')$ choose. That is $x_i((t_i',t'_{-i}))=0$ for any $i\in C^{*}_m(t')\setminus C^{*}_g(t')$ and any $i$'s strategy $t_i'$. In a word, we conclude that reporting truthfully is a dominant strategy for every node.
\end{proof}
Next, we show that a CDM dominates the Vickrey auction, meaning that both the allocation efficiency and the seller's revenue are no less than that given in the Vickrey auction.
\begin{theorem}\label{cdm_do}
	The critical diffusion mechanism dominates the Vickrey auction.
\end{theorem}
\begin{proof}
	Since only nodes in $C^{*}_g(t)$ could have non-zero payments, then $Rev(t)=\sum_{i\in C^{*}_g(t)}x_i(t)$
	$=W^*(t_{-1})+\sum_{i\in C_g(t)\setminus\{1,g\}}({W^*(t_{-{i+1}})-W^*(t_{-\alpha_i})})\geq W^*(t_{-1})$. The inequity holds because of the first property of $\alpha_i$ which means $W^*(t_{-j})\geq W^*(t_{-\alpha_i})$ for any $j\in d_{i+1}$.
	
	Because there exists at most one node's type that belongs to $\{t_i\}_{i\in r_s}\cap \{t_{1}\}$, then $W^*(t_{-1})$ is at least the second highest bid in $\{t_i\}_{i\in r_s}$. In addition, due to the fact that $\{t_i\}_{i\in r_s}\subseteq t_{-\alpha_g}$ and $v_g=W^*(t_{-\alpha_g})$, we conclude that any critical diffusion mechanism dominates the Vickrey auction.
\end{proof}
Essentially, different choices of $\alpha_i$ offer different tradeoffs between the allocation efficiency and the revenue.
The information diffusion mechanism (IDM) proposed in \cite{li2017mechanism} is equivalent to a specific instantiation of the above CDM in which $\alpha_i$ is set to be the edge set $\{(j,i+1)\in E\}$ for any $i\in C^*_m(t)\setminus\{m\}$. According to different preferences of the mechanism designer, one can conveniently generate many such mechanisms. For example, here is a another specific instantiation of CDM: let $\beta_m=\emptyset$ and for any $i\in C^*_m(t)\setminus\{m\}$, $\beta_i$ is the minimum subset of $\{(i,j)\}_{j\in r_i}$, by cutting which the sale information cannot reach node $i+1$.

In Theorem \ref{cdm_outcome}, we proved that given any type profile $t$, the lower bound of the seller's revenue in any CDM is $W^*(t_{-1})$. This is exactly what the seller obtains in the IDM. Therefore, we directly get the following corollary.
\begin{corollary}
	In all critical diffusion mechanisms defined in Alg. \ref{CSM}, the information diffusion mechanism proposed in \cite{li2017mechanism} is the one with the lowest revenue.
\end{corollary}
To give an intuitive description of CDM, we give a running example with respect to $\beta_i$ in Figure \ref{summary}. Firstly, G is the highest bidder and $C^*_G$ is $\{B,F,G\}$ by Def. \ref{critical diffusion nodes}. Then, let $i=B$, if she does not inform of the sale to D nor E, then $i+1=F$ cannot join in the auction. Note that in Figure \ref{summary}, $\beta_B=\{(B,D),(B,E)\}$ is one, and the minimum one such edge cut set. B loses the item since C is the highest bidder after removing $\beta_B$ from the graph. If B is removed from the graph, the highest bidder becomes A. According to the last line in Alg. \ref{CSM}, B's payment is $v_A-v_C=1-4=-3$, i.e., the seller pays 3 to B. In a similar way, $\beta_F=\{(F,G),(F,H)\}$. Since F is the highest bidder after deleting $\beta_F$, then F wins the item and the algorithm stops here. Since E is the highest if F is removed, then according to lines 6-8 in Alg. \ref{CSM}, F pays $v_E=6$. Others pay zero. Finally the seller's revenue is $-3+6=3$ which is greater than 1$-$the revenue in the Vickrey auction.

\section{Auction Mechanism on Weighted Graph}
This section studies auction mechanisms on general weighted graphs. There are two key differences make the problem challenging. On the one hand, $L^*_m(t)$ (or $\pi^*(t)$) could be any simple path from the seller $s$ to $m$ on an unweighted graph while $C^*_m(t)$ is a set of some cut nodes which may or may not be adjacent. For a node $i$ in $L^*_m(t)$ but not in $C^*_m(t)$ (for instance, E in Figure 1), no matter what her strategy $t'_i$ is, she cannot be in $C^*_m((t'_i,t_{-i}))$ since $m\in t_{-i}$.
However, on a weighted graph, since the valuations (bids) are mingled with edge weights, such nodes could affect the efficient allocation by strategic diffusion.
One the other hand, under the unweighted settings, only nodes in $C^{*}_m(t)$ are considered as the candidates of the winner. Nonetheless, nodes in $L^*_m(t)\setminus C^{*}_m(t)$ cannot be omitted on weighted graphs. For example, in Figure 1, if the weight of edge $(D,F)$ is $5$, then if $E$ does not diffuse to $F$, then she becomes the winner (also a critical diffusion node) in the auction.
In this section, we propose the very first mechanism that satisfies all the desired properties on general weighted graphs.

\begin{defn}\label{inter_nodes}
	For any node $i$ and its neighbor $j$, if $j$ has neighbor $k\in r_j(k\neq i)$, then $j$ is an intermediary of $i$.
\end{defn}
\begin{algorithm}[t]
	\SetKwInOut{Input}{\textbf{Input}}\SetKwInOut{Output}{\textbf{Output}}
	\BlankLine
	initialize $\pi(t')=\emptyset$\;
	compute $\pi^*(t')$,  break tie arbitrarily\;
	denote $\pi^*(t')$ by $L_m^*(t')=\{1^*,2^*,\cdots,q^*=m\}$\;
	\For{$i\leftarrow 1^*$ \KwTo $q^*$}{
		compute $\gamma_i$\;
		\If{$i$ is allocated the item in $\pi^*(t'_{-\gamma_{i}})$}{
			set $\pi(t')=L^*_i(t')$\; 
			break\;
		}
	}
	\caption{The allocation policy of WDM\label{wdm_all}}
\end{algorithm}
For instance, in Figure 1, $B$'s intermediaries include nodes $A$, $D$ and $E$. Denote $I_{i}$ by the set of $i$'s intermediaries. Given a reported type profile $t'$, let $m$ be the winner in $\pi^*(t')$ where $\pi^*(t')=L^*_m(t')=\{1^*,2^*,\cdots,q^*=m\}$. Since the graph contains no negative cycles, then according to the optimal substructure property of the shortest path \cite{bondy1976graph}, for any $i^*\in L^*_m(t')$ we have $L^*_{i^*}(t')=\{1^*,2^*,\cdots,i^*\}$.
\begin{defn}\label{gamma}
	Given a reported type profile $t'$ and $\pi^*(t')$, define $\gamma_m=\emptyset$ and for any $i^*\in L^*_m(t')\setminus\{m\}$, let $\gamma_{i^*}$ be the edge set $\{(i^*,j)|j\in I_{i^*}\cup \{(i+1)^*\}\}$.
\end{defn}
Intuitively, $\gamma_{i^*}$ is an ``edge cut set" following $i^*$ since once  $\gamma_{i^*}$ is removed from the graph, nodes in $L^*_m(t')\setminus L^*_{i^*}(t')$ can only be reached via other nodes but not $i^*$.
In Figure \ref{summary}, $\pi^*(t')=\{B,E,F\}$ and $\gamma_{B}=\{(B,i)|i\in \{A,D,E\}\}$. When the edge set $\gamma_B$ is removed from Figure \ref{summary}, all simple paths to node $E$ or $F$, if any, do not pass $B$.

Now we propose the weighted diffusion mechanism (WDM) for general weighted graphs. The allocation policy is given in Alg. \ref{wdm_all}.
In the WDM, we allocate the item to a node $i$ along $L^*_m(t')$. Node $i$ should be the first node which satisfies $v'_i-\sum_{l\in L^*_i(t'_{-\gamma_i})\setminus \{i\}}w(l,l+1)=W^*(t'_{-\gamma_i})$. That is, after removing $\gamma_i$, she is allocated with the item in the efficient allocation with respect to the remaining graph. Note that with such an allocation rule, the nodes who are not in $C^{*}_m(t')$ could also win the auction. We give a running example in Figure \ref{summary}. Firstly, identify the efficient allocation path $B\rightarrow E\rightarrow F$. Since it is node $C$ instead of $B$ who wins in $\pi^*(t_{-\gamma_B})$, we move to check node $E$, but $E$ also loses the auction. Thus finally, node $F$ wins the item.

Before precisely characterizing the payment policy of the WDM, we need another important concept below.

\begin{defn}\label{O_nodes}
	For node $i \in L^*_g(t')\setminus\{g\}$ where $g$ is the winner, we call her a secondary node if
	$i$ wins the item in $\pi^*(\{t'_{-\gamma_i}\setminus t_g'\}\cup (v_g'=nil, r'_g))$ where $g$'s reported type is replaced by $(nil,r_g')$ in $t'_{-\gamma_i}$.
\end{defn}
If node $i$ is a secondary node, then $i$ will be the winner in $\pi^*(t'_{-\gamma_i})$ when $g$ does not bid ($v_g'=nil$). Furthermore, if winner $g\neq m$ beats node $i$ in the efficient allocation $\pi^*(t'_{-\gamma_i})$, so will node $m$ and this violates the definition of secondary nodes. Therefore, the secondary nodes could exist only when the winner is $m$.
Since we search for the winner from head to tail along $L^*_m(t')$, it is important that these secondary nodes are winner's critical opponents. The full characterization of the payment policy is given in Alg. \ref{wdm_pay}.

\begin{algorithm}[t]
	\SetKwInOut{Input}{\textbf{Input}}\SetKwInOut{Output}{\textbf{Output}}
	
	\BlankLine
	initialize $\{x_i(t')=0\}_{i\in V\setminus\{s\}}$ and $B_g^*(t')=0$\;
	let $L^*_g(t')$ be the allocation achieved in Alg. \ref{wdm_all}\;
	denote $\tilde{w}(i,j)=\sum_{L^*_j(t'_{-\gamma_i})\setminus\{j\}}w(l,l+1)$\;
	\For{$i\in L^*_g(t')\setminus\{g\}$}{
		compute $\gamma_i$\;
		set $x_{i}(t')=W^*(t'_{-{i}})-W^*(t'_{-\gamma_{i}})$\;
		\If{$i$ is a secondary node}{
			update $B_g^*(t')=\max\{B_g^*(t'),v'_i-\tilde{w}(i,i)+\tilde{w}(i,g)$\}\;
		}
	}
    update $B_g^*(t')=\max\{B_g^*(t'),W^*(t'_{-g})+\tilde{w}(g,g)\}$\;
	set $x_g(t')=B_g^*(t')$\;
	\caption{The payment policy of WDM\label{wdm_pay}}
\end{algorithm}
%


The value $B_g^*(t')$ computed in Alg. \ref{wdm_pay} can be explained as the ``critical value" of the winner in our diffusion model.
There are two criteria that the ``critical value" should satisfy if $g$ wants to win the item.
Firstly, since $g$ wins the item in $\pi^*(t'_{-\gamma_g})$, for any $i \in -\gamma_g (i\neq g)$, the inequality $v_g'-\tilde{w}(g,g)>v_i'-\tilde{w}(g,i)$ must hold, where $\tilde{w}(i,j)$ denotes the total weights of the shortest trading path to $j$ with respect to $t'_{-\gamma_i}$. To ensure that $g$ wins, her bid $v_g'$ is at least
$$\max_{i \in -\gamma_g (i\neq g)}\{v_i'-\tilde{w}(g,i)+\tilde{w}(g,g)\}.$$ Note that  $\tilde{w}(g,i)$ and $\tilde{w}(g,g)$ are independent of $r_g'$ and the node set with non-nil types in $-\gamma_g$ is minimized when $r_g'=\emptyset$, therefore $\max_{i \in -\gamma_g (i\neq g)}\{v_i'-\tilde{w}(g,i)+\tilde{w}(g,g)\}$ is minimized when $g$ diffuses the sale information to no one, in which case $\max_{i \in -\gamma_g (i\neq g)}\{v_i'-\tilde{w}(g,i)\}=W^*(t'_{-g})$. That is, node $g$ could win the item by reporting $(v_g^{'1}=W^*(t'_{-g})+\tilde{w}(g,g),r_g'=\emptyset)$.
On the other hand, due to the fact that all nodes in $L^*_g(t')\setminus \{g\}$ have priorities to win the item, winner $g$ has to beat all of them one after another. Specifically, given a node $i\in L^*_g(t')\setminus\{g\}$ and a reported type profile $t'_{-\gamma_i}$, if node $i$ is not a secondary node, then it suffices for $g$ to beat $i$ by bidding $v_g^{'1}$. If node $i$ is a secondary node, then according to Alg. \ref{wdm_all}, line 6, the winner must beat node $i$ in the efficient allocation $\pi^*(t'_{-\gamma_i})$. That is to say, we must have $v_g'-\tilde{w}(i,g)> v'_i-\tilde{w}(i,i)$ which leads to $v_g'> v'_i-\tilde{w}(i,i)+\tilde{w}(i,g)$. Denote the set of all secondary nodes by $SN^*$, above argument means that $v_g'$ is at least
$$\max_{i\in SN^*}\{v'_i-\tilde{w}(i,i)+\tilde{w}(i,g)\}.$$ It is not hard to see that the cardinality of $SN^*$ is minimized when $g$ spreads the sale information to all her neighbors, and therefore node $g$ could beat all secondary nodes by reporting a type with bid $v_g^{'2}=\max_{i\in SN^*}\{v'_i-\tilde{w}(i,i)+\tilde{w}(i,g)\}$ and diffusion $r_g'=r_g$.
Therefore, according to the above analysis, winner $g$ can win the item by bidding at least $\max\{v_g^{'1},v_g^{'2}\}$.

In Figure \ref{summary}, regarding whether node $E$ wins the item in the efficient allocation after cutting $\gamma_E=\{(E,B),(E,F)\}$, since node $E$ would win the item in $\pi^*(t'_{-\gamma_E})$ if the winner $F$ does not bid, therefore, node $E$ is a secondary node (actually the only one). According to Alg. \ref{wdm_pay}, winner $F$ will pay a price of at least $v_E-\tilde{w}(E,E)+\tilde{w}(E,F)=v_E-w(B,E)+(w(B,D)+w(D,F))=6-0+3=9$. Since $W^*(t_{-F})+\tilde{w}(F,F)=6$, then the winner pays $9$ to the seller. One could further check that node $B$ and $E$ pay $-2$ and $0$ respectively. The seller's revenue is $-2+0+9-0=7$ which is greater than $1$, the revenue obtained in the Vickrey auction.

The identification of secondary nodes and the characterization of $B^*_g(t')$ are the key techniques in proving the property of incentive compatibility. Combining with the following two lemmas, we prove that WDM is IR and IC.

\begin{lemma}\label{O_winner}
	If there exist secondary nodes in Alg. \ref{wdm_all}, then the winner cannot increase her utilities by misreporting.
\end{lemma}
\begin{proof}
	First of all, the winner must be $m$ if there exist secondary nodes according to Def. \ref{O_nodes}. For any secondary node $j$, she is still a secondary node for any $m$'s diffusion strategy $r_m'\subset r_m$ as long as $m$ remains to be the winner. Therefore $m$'s payment is minimized by diffusing the sale information to all her neighbors. Since there exist secondary nodes, then for any secondary node $j$, the winner $m$ is the only node that beats $j$ in $\pi^*(t_{-\gamma_j})$. This leads to the fact that when $m$ is trying to become a node along the winning path by lowering her bid, the first secondary node in type profile $t$ will become the winner in advance according to Alg. \ref{wdm_all}, line 6, whereby $m$'s utility becomes zero. Therefore, no matter what strategies $m$ commits, she cannot increase her utility.
\end{proof}

\begin{lemma}\label{kick_off_nodes}
	Given any type profile $t$ and any node $i\in L^*_m(t)\setminus L^*_g(t)$, $i$ cannot increase her utility by misreporting.
\end{lemma}
\begin{proof}
	Firstly, when $i$ spreads the sale information to only a part of her neighbors, there are two possible results: she is still an node in $L^*_m(t)$ in which case $g$ still wins according to Def. \ref{inter_nodes} and Alg. \ref{wdm_all}; or she is out of $L^*_m(t)$ which makes her out of $L^*_g(t)$ either. Therefore, her utility remains unchanged by misreporting $r_i$. Secondly, note that unless node $i$ beats $g$ in $\pi^*(t_{-\gamma_g})$, otherwise $g$ must be the winner and $i$'s utility remains zero according to Alg. \ref{wdm_pay}. Since $W^*(t_{-\gamma_g})$ is greater than $v_j-\tilde{w}(g,j)$ for any $j\in L^*_m(t)\setminus L^*_g(t)$, she can beat $g$ only by increasing her bid such that she wins in Alg. \ref{wdm_all}. Due to the fact that node $g$ is the winner in $\pi^*(t_{-\gamma_g})$, node $g$ would be a secondary node if $i$ becomes the winner. In this case, node $i$ has to pay at least $v_g-\tilde{w}(g,g)+\tilde{w}(g,i)=W^*(t_{-\gamma_g})+\tilde{w}(g,i)$. Therefore, the utility of $i$ by being a winner is at most $v_i-(W^*(t_{-\gamma_g})+\tilde{w}(g,i))=(v_i-\tilde{w}(g,i))-W^*(t_{-\gamma_g})<0$.
\end{proof}
\begin{theorem}\label{mlma_outcome}
	The weighted diffusion mechanism is individually rational and incentive-compatible.
\end{theorem}
\begin{proof}
	Firstly, we show that the weighted diffusion mechanism is individually rational. For any node $i\notin L^*_g(t')$, her payment is zero. For any node $i\in L^*_g(t')\setminus\{g\}$, her utility is $W^*(t'_{-\gamma_i})-W^*(t'_{-i})\geq 0$ since $t'_{-i}\subset t'_{-\gamma_i}$. As for the winner $g$, her payment is $B^*_g(t')$ which is no more than $v_g$.
	
	Next we show WDM is IC through following cases.
	
	Case 1: For $i\notin L^*_m(t)$, since $\{t_j\}_{j\in L^*_m(t)}\subseteq t_{-i}$, the only chance for $i$ to have a non-zero utility is to bid high enough such that she becomes the winner in $\pi^*(t')$. And in this case she has to pay a price of at least $W^*(t)+\tilde{w}(i,i)$ which is higher than her true value $v_i$ because $W^*(t)>v_i-\tilde{w}(i,i)$.
	
	Case 2: For $i\in L^*_g(t)\setminus\{g\}$, her utility is $W^*(t_{-\gamma_i})-W^*(t_{-i})$. Since $W^*(t_{-\gamma_i})\ge v_i-\tilde{w}(i,i)$, she cannot do better by winning, otherwise she pays at least $W^*(t_{-i})+\tilde{w}(i,i)$.
	The latter term $W^*(t_{-i})$ is independent of $i$'s strategy and the former term $W^*(t_{-\gamma_i})$ is independent of $v_i$ (otherwise she becomes the winner according to Alg. \ref{wdm_all}, line 6). Due to Def. \ref{inter_nodes}, for any $r_i'\subset r_i$ $I^*_i\subseteq I^{*'}_i\cup r_i\setminus r_i'$. This property induces that $\gamma_i\subseteq \gamma_i'\cup r_i\setminus r_i'$. That is, after removing $\gamma_i$ and $\gamma_i'$ from the graph separately, the set of non-nil type nodes follows that $-\gamma_i'\subseteq -\gamma_i$.
	Therefore the set of non-nil types in $t_{-\gamma_i}$ is the largest when $i$ spreads the sale information to all her neighbors, this maximizes her utility.
	
	Case 3: For winner $g$, her utility is $v_g-B^*_g(t)$.
	If there are no secondary nodes, then she pays exactly $W^*(t_{-g})+\tilde{w}(g,g)$ and her utility is $v_g-\tilde{w}(g,g)-W^*(t_{-g})=W^*(t_{-\gamma_g})-W^*(t_{-g})$. When she misreports to be in Case 2, her utility becomes $W^*(t_{-\gamma'_g})-W^*(t_{-g})$. Since the set of non-nil types in $-\gamma_g$ is maximized when $g$ diffuses the sale to all her neighbors, misreporting is a bad choice for her since $W^*(t_{-\gamma_g})\ge W^*(t_{-\gamma'_g})$. On the other hand, the more non-nil nodes exist in the graph, the harder the secondary nodes could exist. Therefore, if $g$ diffuses the sale only to a part of her neighbors, then secondary nodes would be created which in turn will increase her payment according to Alg. \ref{wdm_pay}, lines 4-8. Hence, truthfully reporting is the best strategy for the winner. In addition, if there exist secondary nodes, then misreporting will harm $g$'s profit according to Lemma \ref{O_winner}.

	Case 4: For any node $i$ in $L^*_m(t)\setminus L^*_g(t)$, according to Lemma \ref{kick_off_nodes}, misreporting would decrease her utility.
	
\end{proof}

For proving that WDM dominants the Vickrey auction, we firstly pick out the nodes with zero payments.

\begin{lemma}\label{zero_x}
	For $i\in L^*_g(t)\setminus\{g\}$, if the intersection of $\pi^*(t_{-\gamma_i})$ and $L^*_m(t)\setminus L^*_i(t)$ is not empty, then $i$'s payment is zero.
\end{lemma}
\begin{proof}
	According to Def. \ref{gamma}, if we remove $\gamma_i$ from the graph, then all nodes in $L^*_m(t)\setminus L^*_i(t)$ can only be reached by other nodes but not $i$. Therefore, if $\pi^*(t_{-\gamma_i})$ and $L^*_m(t)\setminus L^*_i(t)$ have some common elements, according to the optimal substructure property of shortest path, node $m$ must be the winner in $\pi^*(t_{-\gamma_i})$ and node $i\notin \pi^*(t_{-\gamma_i})$. Consequently, we have that $W^*(t_{-\gamma_i})$ must equal $W^*(t_{-i})$ and therefore $x_i(t)=0$.
\end{proof}
For convenience, denote $L_g^\#(t)$ by the nodes in $L_g^*(t')$ who satisfy $\pi^*(t_{-\gamma_i})\cap L^*_m(t)\setminus L^*_i(t)=\emptyset$, then according to Lemma \ref{zero_x} only $L^\#_g(t)$ can have non-zero payments.
\begin{theorem}\label{bdm_bb}
The weighted diffusion mechanism dominates the Vickrey auction.
\end{theorem}
\begin{proof}
According to the payment policy of WDM and Lemma \ref{zero_x}, the seller's revenue can be denoted as
	\begin{small}
		\begin{align*}
		Rev(t) &=\sum_{j\in V\setminus\{s\}}{x_j(t)}-\sum_{l\in L^*_g(t)\setminus\{g\}}{w(l,l+1)}\\
		&=\sum_{j\in L^\#_g(t)}{x_j(t)}-\tilde{w}(g,g)\\
		&=\sum_{j\in L^\#_g(t)\setminus\{g\} }(W^*(t_{-j})-W^*(t_{-\gamma_j}))\\
		&+B^*_g(t)-\tilde{w}(g,g)\\
		&\geq \sum_{j\in L^\#_g(t)\setminus\{g\} }(W^*(t_{-j})-W^*(t_{-\gamma_j}))+W^*(t_{-g})\\
		&=W^*(t_{-1^\#})+\sum_{j\in L^\#_g(t)\setminus\{1^\#\} }(W^*(t_{-{j}})-W^*(t_{-\gamma_{j-1}}))\\
		&\geq W^*(t_{-1^\#}).
		\end{align*}
	\end{small}
	The first inequation holds because $B^*_g(t)$ is at least $W^*(t_{-g})+\tilde{w}(g,g)$ according to Alg. \ref{wdm_pay}, line 9. Since for any $j\in L_g^\#(t)$, we have that $\pi^*(t_{-\gamma_j})\cap L^*_m(t)\setminus L^*_i(t)=\emptyset$. Then we have $W^*(t_{-{j}})\geq W^*(t_{-\gamma_{j-1}})$ for any $j\in L_g^\#(t)\setminus \{1^\#\}$ which induces the second inequity.

	Because $\{t_i\}_{i\in r_s}\subseteq t_{-1^\#}\cup \{t_{1^\#}\}$, we have that $W^*(t_{-1^\#})$ is at least the second highest bid in $r_s$. In addition, since $v_g-\tilde{w}(g,g)\geq W^*(t_{-\gamma_{1^*}})$ and $\{t_i\}_{i\in r_s}\subseteq t_{-\gamma_{1^*}}$, then the allocation efficiency of the weighted diffusion mechanism is no less than that given in the Vickrey auction.
\end{proof}

It is worth noting that when the weights are zero, WDM degenerates to a CDM.
\section{Conclusions}
In this paper, we formulate the model of information diffusion and auction on both unweigted and weighted graphs and provide generic solutions which not only guarantee IC and IR but also can incentivize information sharing. The key techniques for realizing these mechanisms are graph cut analysis and the integration of bidder's private valuation and social links. Due to the introduction of information sharing incentive, all the instance mechanisms under our solution framework improve both the seller's revenue and the allocation efficiency comparing with the Vickrey auction. In practice, these mechanisms can be applied in a recursive way: any new comer will be treated as a bidder if she submits a two-dimensional type. This recursion is terminated if there is no new submission. Then the seller calculates the trading path, and determines the allocation and payments.
One immediate future work is: although this work realizes seller's revenue optimization, the problem of its maximization is still unclear. Moreover, what is the underlying impact of $\alpha_i$ on the seller's revenue and the allocation efficiency, and how to design fast algorithms determining the membership of these sets are also important.



\section*{Acknowledgments}
This work was supported by the National Natural Science Foundation of China (NNSFC) Grant Number
71601029 and JSPS KAKENHI Grant Number JP17H00761.
\bibliographystyle{named}
\bibliography{ijcai19}

\begin{thebibliography}{}

\bibitem[\protect\citeauthoryear{Belloni \bgroup \em et al.\egroup
  }{2017}]{belloni2017mechanism}
Alexandre Belloni, Changrong Deng, and Sa{\v{s}}a Peke{\v{c}}.
\newblock Mechanism and network design with private negative externalities.
\newblock {\em Operations Research}, 65(3):577--594, 2017.

\bibitem[\protect\citeauthoryear{Bhattacharya \bgroup \em et al.\egroup
  }{2011}]{bhattacharya2011allocations}
Sayan Bhattacharya, Janardhan Kulkarni, Kamesh Munagala, and Xiaoming Xu.
\newblock On allocations with negative externalities.
\newblock In {\em International Workshop on Internet and Network Economics},
  pages 25--36. Springer, 2011.

\bibitem[\protect\citeauthoryear{Bondy \bgroup \em et al.\egroup
  }{1976}]{bondy1976graph}
John~Adrian Bondy, Uppaluri Siva~Ramachandra Murty, et~al.
\newblock {\em Graph theory with applications}, volume 290.
\newblock Citeseer, 1976.

\bibitem[\protect\citeauthoryear{Bulow and Klemperer}{1996}]{Bulow1996Auction}
Jeremy Bulow and Paul Klemperer.
\newblock Auction versus negotiations.
\newblock {\em American Economic Review}, 86(1):180--94, 1996.

\bibitem[\protect\citeauthoryear{{Clarke}}{1971}]{clarke1971multipart}
Edward~H. {Clarke}.
\newblock Multipart pricing of public goods.
\newblock {\em Public Choice}, 11(1):17--33, 1971.

\bibitem[\protect\citeauthoryear{Cramton}{2013}]{cramton2013spectrum}
Peter Cramton.
\newblock Spectrum auction design.
\newblock {\em Review of Industrial Organization}, 42(2):161--190, 2013.

\bibitem[\protect\citeauthoryear{Edelman \bgroup \em et al.\egroup
  }{2007}]{edelman2007internet}
Benjamin Edelman, Michael Ostrovsky, and Michael Schwarz.
\newblock Internet advertising and the generalized second-price auction:
  Selling billions of dollars worth of keywords.
\newblock {\em American Economic Review}, 97(1):242--259, 2007.

\bibitem[\protect\citeauthoryear{Groves}{1973}]{groves1973incentives}
Theodore Groves.
\newblock Incentives in teams.
\newblock {\em Econometrica: Journal of the Econometric Society}, pages
  617--631, 1973.

\bibitem[\protect\citeauthoryear{Haghpanah \bgroup \em et al.\egroup
  }{2013}]{haghpanah2013optimal}
Nima Haghpanah, Nicole Immorlica, Vahab Mirrokni, and Kamesh Munagala.
\newblock Optimal auctions with positive network externalities.
\newblock {\em ACM Transactions on Economics and Computation}, 1(2):13, 2013.

\bibitem[\protect\citeauthoryear{Krishna}{2009}]{krishna2009auction}
Vijay Krishna.
\newblock {\em Auction theory}.
\newblock Academic Press, 2009.

\bibitem[\protect\citeauthoryear{Li \bgroup \em et al.\egroup
  }{2017}]{li2017mechanism}
Bin Li, Dong Hao, Dengji Zhao, and Tao Zhou.
\newblock Mechanism design in social networks.
\newblock In {\em AAAI}, pages 586--592, 2017.

\bibitem[\protect\citeauthoryear{Li \bgroup \em et al.\egroup
  }{2018}]{Li2018CustomerSI}
Bin Li, Dong Hao, Dengji Zhao, and Tao Zhou.
\newblock Customer sharing in economic networks with costs.
\newblock In {\em IJCAI}, pages 368--374, 2018.

\bibitem[\protect\citeauthoryear{Myerson}{1981}]{Myerson1981OptimalAD}
Roger~B. Myerson.
\newblock Optimal auction design.
\newblock {\em Mathematics of Operations Research}, 6:58--73, 1981.

\bibitem[\protect\citeauthoryear{{Rothkopf}}{2007}]{rothkopf2007thirteen}
Michael~H. {Rothkopf}.
\newblock Thirteen reasons why the vickrey-clarke-groves process is not
  practical.
\newblock {\em Operations Research}, 55(2):191--197, 2007.

\bibitem[\protect\citeauthoryear{Vickrey}{1961}]{vickrey1961counterspeculation}
William Vickrey.
\newblock Counterspeculation, auctions, and competitive sealed tenders.
\newblock {\em The Journal of Finance}, 16(1):8--37, 1961.

\bibitem[\protect\citeauthoryear{Zhao \bgroup \em et al.\egroup
  }{2018}]{Zhao2018Multi}
Dengji Zhao, Bin Li, Junping Xu, Dong Hao, and Nicholas~R. Jennings.
\newblock Selling multiple items via social networks.
\newblock In {\em AAMAS}, pages 68--76, 2018.

\end{thebibliography}

\end{document}